\documentclass[11pt]{article}
\usepackage{amsmath,amsthm,amssymb,hyperref, color}
\usepackage{fullpage}

\newtheorem{thm}{Theorem}[section]
\newtheorem{claim}[thm]{Claim}
\newtheorem{lem}[thm]{Lemma}
\newtheorem{define}[thm]{Definition}
\newtheorem{cor}[thm]{Corollary}

\def\F{{\mathbb{F}}}
\def\Q{{\mathbb{Q}}}
\def\Z{{\mathbb{Z}}}

\def\cF{{\cal F}}

\def\cM{{\cal M}}

\def\cR{{\mathcal{R}}}

\newcommand{\ip}[2]{\langle #1,#2 \rangle}

\def\_{\,\,\,\,\,}

\def\rank{\textsf{rank}}

\newcommand{\eps}{\epsilon}

\newcommand{\remove}[1]{}

\begin{document}

\title{Matrix Rigidity and the Croot-Lev-Pach Lemma}

\author{Zeev Dvir\thanks{Department of Computer Science and Department of Mathematics,
Princeton University.
Email: \texttt{zeev.dvir@gmail.com}. Research supported by NSF CAREER award DMS-1451191 and NSF grant CCF-1523816 } \and
Ben Edelman\thanks{Department of Mathematics, Princeton University.
Email: \texttt{benedelman@princeton.edu}.}}

\date{}
\maketitle

\begin{abstract}
Matrix rigidity is a notion put forth by Valiant \cite{Val77} as a means for proving arithmetic circuit lower bounds. A matrix is rigid if it is far, in Hamming distance, from any low rank matrix. Despite decades of efforts, no explicit matrix rigid enough to carry out Valiant's plan has been found. Recently, Alman and Williams \cite{AW17} showed, contrary to common belief, that the $2^n \times 2^n$ Hadamard matrix $H_n = \left( (-1)^{\ip{x}{y}} \right)_{x,y \in \F_2^n}$ could not be used for Valiant's program as it is not sufficiently rigid.

In this note we observe a similar `non rigidity' phenomena for {\em any}  $q^n \times q^n$ matrix $M$ of the form $M(x,y) = f(x+y)$, where $f:\F_q^n \to \F_q$ is any function and $\F_q$ is a fixed finite field of $q$ elements ($n$ goes to infinity). The theorem follows almost immediately from a recent lemma of Croot, Lev and Pach \cite{CLP17} which is also the main ingredient in the recent solution of the  cap-set problem \cite{EG17}.
\end{abstract} 

\section{Introduction}

We begin by defining the notion of matrix rigidity -- a property of matrices that combines combinatorial conditions (Hamming distance) with algebraic ones (matrix rank). Recall that the Hamming distance between two vectors $x,y \in \Sigma^n$ over some alphabet $\Sigma$ is equal to the number of entries $i \in [n]$ for which $x_i \neq y_i$.
\begin{define}[Matrix rigidity]
The rank-$r$ rigidity of a matrix $M$ over a field $\F$, denoted $\cR_M^\F(r)$,  is defined as the minimum Hamming distance between $M$ and any matrix of rank at most $r$. In other words,  $\cR_M^\F(r)$ is equal to the smallest number of entries in $M$ that one needs to change in order to reduce the rank of $M$ to $r$.
\end{define}

Specifying the field is important since some integer matrices can have much higher rigidity over the rational numbers than over finite fields (this holds true even if one only considers the rank itself). 

The notion of matrix rigidity was introduced by  Valiant \cite{Val77} in the context of studying  the arithmetic circuit complexity of linear transformations. A \emph{linear circuit} is a model of computation in which the inputs represent the basic linear function $x_1,\ldots,x_n$ and each gate  takes two previously computed linear forms and outputs some linear combination of them with coefficients in the field. We measure the size of a linear circuit by counting the number of wires, and the depth by the longest path from input to output. A linear circuit with $n$ inputs and $n$ outputs computes a linear map $T: \F^n \mapsto \F^n$ and many important linear maps (e.g., Fourier transform) can be computed efficiently in this model.  One can even show that any use of multiplication gates can be eliminated (with negligible cost) when computing a linear map \cite{Lok09}. 

One of the most important problems in theoretical computer science is to prove unconditional complexity lower bounds for realistic models of computation. Despite decades of attempts, we are still unable to prove super-linear circuit lower bounds (in any realistic model) for logarithmic depth circuits. In an early attempt to bridge this gap Valiant proved the following theorem. 
\begin{thm}[Valiant \cite{Val77}]\label{val}
Let $M$ be an $N\times N$ matrix over a field $\F$. If  $$\cR_M^\F(N/\log{\log{N}}) \geq \Omega(N^{1+\varepsilon})$$ for some $\varepsilon > 0$  then $M$ cannot be computed by linear circuits of size $O(N)$ and depth $O(\log(N))$ (asymptotically, as $N$ grows\footnote{To be more precise, one would have to consider the rigidity of an infinite sequence of matrices indexed by $N$.}). 
\end{thm}

We can say that a matrix is `Valiant-rigid' if it satisfies the rigidity parameters in the above theorem. It is straightforward to check that for any matrix $M$ and field $\F$, $\mathcal{R}_M^\F(r) \leq (N-r)^2$ for any $r$. Valiant proved that almost all matrices achieve this maximum rigidity: for almost all matrices $M$, $\mathcal{R}_M^\F(r) = (N-r)^2$ if $\F$ is infinite and $\mathcal{R}_M^\F(r) = \Omega((N-r)^2/\log{N})$ if $\F$ is finite. However, since Valiant's original paper, it remains an open problem to find an explicit `Valiant-rigid' matrix. By `explicit' we mean a matrix that can be produced in polynomial (in $N$) time by a Turing machine given $N$ as input.

The current best rigidity lower bound for any explicit matrix is $\mathcal{R}_{M}^F(r) = \Omega(\frac{N^2}{r}\log{\frac{N}{r}})$ \cite{Fri93,SSS97}. Until recently, the $2^n \times 2^n$ Hadamard matrix $H = ((-1)^{\langle x,y\rangle})_{x,y \in \{0,1\}^n}$ was conjectured to be Valiant-rigid over the rational numbers \cite{Lok09}. A recent surprising result of Alman and Williams   \cite{AW17} showed that in fact the Hadamard matrix is not sufficiently rigid. Denoting $N = 2^n$,  they showed that for every $\eps>0$ there exists $\eps'>\Omega(\eps^2/\log(1/\eps))$ such that  $\mathcal{R}_H^\Q(N^{1-\eps'}) \leq N^{1+\eps}$. 

The purpose of this note is to observe another `non-rigidity' phenomenon for a related (large) family of matrices. The hope is that by understanding the reasons for this non-rigidity we can perhaps get closer to proving stronger rigidity results. Our main theorem is the following.
\begin{thm}\label{thm-main}
Let $\F_q$ be any finite field and let $f: \F_q^n \to \F_q$ be any function. Let $M$ be the $q^n \times q^n$ matrix defined by $M_{x,y} = f(x+y)$ for $x,y \in \F_q^n$. Denoting $N = q^n$ we have that for any $\eps>0$, there exists $\eps'>0$ such that $\mathcal{R}_M^{\F_q}(N^{1-\eps'}) \leq N^{1+\eps}$. The result holds for fixed $q$ and $\eps$ and $n$ sufficiently large.\end{thm}

One should note that, unlike the Hadamard matrix, these matrices are over a finite field and not over the rational numbers. Having non-rigid matrices over a finite field is a bit less surprising since there are more `ways' for the rank to be low. It is an interesting open problem to determine if Theorem~\ref{thm-main} still holds if one is allowed to take a function $f: \F_q^n \mapsto \F$ where $\F$ is the rational numbers (or even the complex numbers). This will imply the results of \cite{AW17} since   the Hadamard matrix  can  be  written over the complex numbers as  $$(-1)^{\langle x,y\rangle} = (-1)^{|x|/2}(-1)^{|y|/2}(-1)^{|x \oplus y|/2},$$ where $|\cdot|$ represents the Hamming weight.

Another interesting question is that of replacing the group $\F_q^n$ indexing the rows/columns with other groups. For example, taking $N \times N$ matrices with entries $f(x+y)$ but with $f: \Z/N\Z \mapsto \F$ an arbitrary function. Here one might expect to see higher rigidity since there are far fewer low rank matrices of this form (c.f, the recent work of Goldreich and Tal \cite{GT15} on the rigidity of Toeplitz matrices). 

\subsection{The Croot-Lev-Pach (CLP) lemma}
A \emph{cap set} is a subset of $\F_q^n$ with no non-trivial three-term arithmetic progressions. We think of $q>2$ as fixed and $n$ going to infinity. The cap set problem asks how the size of the largest possible cap set (denoted $r(n)$) grows in terms of $n$. It was an open question whether $r(n) \leq c^n$ for some $c < q$. Croot, Lev, and Pach \cite{CLP17} used a variant of the polynomial method to solve the corresponding problem for $\Z_4^n$  (the ring mod 4) in the affirmative, proving a bound of $c^n$ for some $c<4$, and soon afterwards Ellenberg and Gijswijt \cite{EG17} adapted the CLP result to provide a positive answer to the cap set problem in $\F_q$ for all $q>2$. At the core of \cite{CLP17} is a lemma saying that, if $P: \F_q^n \to \F_q$ is a polynomial of not too high degree, then the $q^n \times q^n$ matrix $M = (P(x+y))_{x,y \in \F_q^n}$ has very low rank (see below for the exact parameters). We observe that, since {\em any} function can be well approximated by such a polynomial, the matrix $f(x+y)$ can be changed in a small number of entries to give the low rank matrix $P(x+y)$.

\section{Proof of Theorem~\ref{thm-main}}

Let $\cF(q,n)$ denote the set of functions $f: \F_q^n \mapsto \F_q$.  Then, $\cF(q,n)$ is an $\F_q$-vector space of dimension $q^n$. A basis for this vector space is given by the set of $q^n$ monomials $$\cM(q,n) = \{ x_1^{a_1}\cdots x_n^{a_n} \,|\, 0\leq a_i \leq q-1 \}.$$ Let us denote by $\cM_d(q,n)$ the set of monomials in $\cM(q,n)$ of total degree  at most $d$ and by $\cF_d(q,n)$ the set of polynomials of degree at most $d$ spanned by these monomials. Let $m_d(q,n)$ denote the size of $\cM_d(q,n)$ or equivalently the dimension of $\cF_d(q,n)$. 

We start by stating the precise form of the CLP lemma. For completeness we include a short sketch of the proof.
\begin{lem}[CLP lemma \cite{CLP17}]\label{lem-clp}
	Let $P \in \cF_d(q,n)$ and let $M$ denote the $q^n \times q^n$ matrix with entries $M_{x,y} = P(x+y)$ for $x,y \in \F_q^n$. Then $\rank(M) \leq 2\cdot m_{\lfloor d/2 \rfloor}(q,n).$
\end{lem}
\begin{proof}[Proof sketch]
	To prove the claim we will show that $P(x+y) = \sum_{i=1}^R f_i(x)g_i(y)$ with $R \leq 2\cdot m_{\lfloor d/2 \rfloor}(q,n)$. To see how to do this observe that, for each monomial $m(x) = x_1^{a_1}\cdots x_n^{a_n}$ of degree at most $d$, the terms in the expression $m(x+y)$ all have degree $\leq \lfloor d/2 \rfloor$ in either $x$ or $y$. Writing $P$ as a sum of monomials and grouping together terms with the same low degree parts (in $x$ first and then in $y$) gives the desired decomposition.
\end{proof}

The main power of the CLP lemma comes from the following quantitative observation. For a fixed $q$ and sufficiently large $n$, the numbers $m_d(q,n)$ behave approximately like a normal curve when we increase $d$ from $0$ to $(q-1)n$ (the largest possible degree). Most of the mass will be concentrated around the middle $(q-1)n/2$ with the tails decaying exponentially fast. We will use the following (weak) estimate.

\begin{claim}\label{cla-largedev}
For any prime power $q$ and any $\eps>0$ there exists $\delta>0$ such that, for sufficiently large $n$, we have 
$$ m_{(1-\delta)(q-1)n} \geq q^n - q^{\eps n}.$$
\end{claim}
\begin{proof}
  By symmetry it is enough to bound $m_{\delta(q-1)n} \leq q^{\eps n}$. We reduce this problem to the binary alphabet case. We claim that $m_d(q,n) \leq m_d(2,n(q-1))$ for all $d$. To see this, consider the injective mapping from $\cM_d(q,n)$ into $\cM_d(2,n(q-1))$ sending $x_i^{a_i}$ to the multilinear monomial $x_{i1}x_{i2}\cdots x_{ia_i}$. For the binary case we can use the standard tail bounds for the Binomial distribution to get that $m_{\delta(q-1)n}(2,n(q-1)) \leq 2^{H(\delta)(q-1)n}$ with $H(\delta)$ going to zero with $\delta$ ($H$ is the binary entropy function). Taking $\delta$ sufficiently small (as a function of $q$ and $\eps$) we can get $2^{H(\delta)(q-1)} \leq q^\eps$, proving the claim.
\end{proof}

The following  claim  and corollary show that any function can be approximated well by a polynomial of sufficiently high degree.

\begin{claim}\label{cla-dim}
Suppose $V$ is a finite-dimensional vector space over a field $\F$ and $W$ is a subspace of $V$. Let $\mathcal{B}$ be a basis for $V$. Then, for any vector $v \in V$, we can modify $\dim{V}-\dim{W}$ of the coordinates of $v$ (in the basis $\mathcal{B}$) to produce a vector that lies in $W$.
\end{claim}
\begin{proof}
Let $n = \dim{V}$ and $r = \dim{W}$. There exists an $n \times r$ rank-$r$ matrix $M$ such that $W$ is the image of the linear transformation given by $M$. We need to show that there is a vector $u$ agreeing with $v$ on $r$ coordinates such that there exists $x \in \F^r$ with $Mx=u$. We can find a size-$r$ subset $S \in [n]$ such that the rows of $M$ indexed by $S$ span all the rows of $M$. For $i \in S$, let $u_i = v_i$. Let $A$ be the $r \times r$ submatrix of $M$ consisting of the rows indexed by $S$. Since $A$ is full rank, there exists exactly one $x \in \F^r$ such that $Ax=(u_i)_{i \in S}$. The other rows of $M$ are spanned by the rows of $A$, so we can choose $u_i$ for each $i \in [n]-S$ by multiplying the $i$th row in $A$ by $x$. Since $Ax = u$, $u \in W$ and we are done.
\end{proof}

\begin{cor}\label{cor-approx}
	Let $f \in \cF(q,n)$ be any function. Then, for all $d \leq n$, there exists a polynomial $P \in \cF_d(q,n)$ that 
	$$ |\{ x \in \F_2^n \,|\, f(x) \neq P(x) \}| \leq q^n - m_d(q,n).$$
\end{cor}
\begin{proof}
	This follows from the previous claim and from the fact that $\dim(\cF_d(q,n)) = m_d(q,n)$.
\end{proof}

We are now ready to prove our main result.
\begin{proof}[Proof of Theorem~\ref{thm-main}]
	Let $f: \F_q \mapsto \F_q$ be as in the theorem and let $\eps>0$. Using Claim~\ref{cla-largedev} and Corollary~\ref{cor-approx}, we can find $\delta>0$ and a polynomial $P$ of degree at most $d = (1-\delta)(q-1)n$ such that $P$ agrees with $f$ on all but $q^{\eps n} = N^\eps$ values in $x \in \F_2^n$. Let $M$ denote the $q^n \times q^n$ matrix with entries $M_{x,y} = f(x+y)$ and let $L$ denote the matrix of the same dimensions with entries $L_{x,y} = P(x+y)$. Then, $M$ and $L$ differ in at most $N^\eps$ entries in each row and in at most $N^{1+\eps}$ entries altogether. Now, by Lemma~\ref{lem-clp} (the CLP lemma) we have that $ \rank(L) \leq m_{\lfloor d/2\rfloor}(q,n)$. But $d/2 = (1/2 - \delta/2)(q-1)n$ and so, by the Chernoff-Hoeffding bound, we have $m_{\lfloor d/2\rfloor}(q,n) \leq q^{(1-\eps')n} = N^{1-\eps'}$ for some $\eps'>0$ depending on $\delta$ (which in turn depends on $q$ and on $\eps$). This concludes the proof.
\end{proof}

\end{document}